\documentclass[a4paper]{paper}
\usepackage[english]{babel}
\usepackage[applemac]{inputenc}
\usepackage{bbm,epigraph}
\usepackage{amsfonts}
\usepackage{amsthm,comment}
\usepackage{color}
\usepackage{graphicx,MnSymbol}
\usepackage[all,cmtip]{xy}
\usepackage{multicol}
\usepackage{tikz}
\usetikzlibrary{arrows}
\usetikzlibrary{calc}
\usetikzlibrary{decorations.markings}
 \usepackage{tikz-cd}

\newcommand{\inp}{	\,\ensuremath{\raise-.5ex\hbox{$\invbackneg$}}} 
\newcommand{\qsp}[2]{\,\ensuremath{\raise.5ex\hbox{$#1$}\big\slash\raise-.5ex\hbox{$#2$}}}

\DeclareMathOperator{\coker}{coker}

\DeclareMathOperator{\val}{val}

\newcommand{\CC}{\mathbb{C}}

\newcommand{\beq}{\begin{equation*}}
\newcommand{\qeb}{\end{equation*}}

\newcommand{\bep}{\begin{prop}}
\newcommand{\peb}{\end{prop}}

\newcommand{\ev}[1]{\textbf{\textsf{#1}}}

\renewcommand{\ker}{\mathsf{ker}\,}

\theoremstyle{definition}

\newtheorem{theorem}{\textsf{Theorem}}

\newtheorem{definizione}{\textsf{Definition}}[section]

\newtheorem{proposition}{\textsf{Proposition}}[section]
\newtheorem{lemma}{\ev{Lemma}}[section]

\newenvironment{definition}{\begin{definizione}}{\newline \phantom{.}\hfill $\maltese$\end{definizione}}

\newtheoremstyle{Rem}% Nome 
{}% Spazio che precede l�enunciato 
{}% Spazio che segue l�enunciato 
{}% Stile del font dell�enunciato 
{}% Rientro (se vuoto, non cՏ rientro, 
{\itshape}% Stile del font dell�intestazione 
{:}% Punteggiatura che segue l�intestazione 
{\newline}% Spazio che segue l�intestazione: 
{}% Specifica l�intestazione dell�enunciato 
\theoremstyle{Rem}
\newtheorem*{Remark}{\textsf{Remark}}

\newenvironment{remark}{\begin{Remark}}{\newline \phantom{.} \hfill $\pentagram$ \end{Remark}}

\title{The Graph Laplacian and Morse Inequalities}
\author{Ivan Contreras\thanks{Department of Mathematics, University of Illinois, Urbana-Champaign, USA. icontrer@illinois.edu.}  - Boyan Xu \thanks{Department of Mathematics, University of Illinois, Urbana-Champaign, USA. borisxu2@illinois.edu.}}
\date{}

\begin{document}
\maketitle

\abstract{{ The objective of this note is to provide an interpretation of the discrete version of Morse inequalities, following Witten's approach via supersymmetric quantum mechanics \cite{W}, adapted to finite graphs, as a particular instance of Morse-Witten theory for cell complexes \cite{FormanCW}. We describe the general framework of graph quantum mechanics and we produce discrete versions of the Hodge theorems and energy cut-offs within this formulation. }}
\keywords{ Graph Laplacian, Morse-Witten complex, graph Hodge theory, discrete Morse functions.}

\tableofcontents

\section*{Introduction}
The understanding of physical phenomena, as well as the behavior of information in networks have been recently studied from the combinatorial and algebraic perspective.\\

In this paper we consider a  toy version of quantum mechanics \cite{delVeccio, Mnev2} based on a graph-theoretic analogue of the Schr\"odinger equation. To a finite graph we associate a \textit{partition function}, a discretization of the Feynman path integral which can be used to count special types of paths on graphs and compute topological invariants. It relies on the discretized version of the Laplace operator $$\Delta= \nabla^2= \sum_{i=1}^{n}\frac{\partial ^2}{\partial x_i^2},$$ which depends on the combinatorics of the graph.  We apply this framework to Morse theory, and we are able to recover Morse inequalities, by following Witten's viewpoint of critical points of Morse functions.\\

In \cite{FormanCW}, Forman provided a combinatorial interpretation via cell complexes of Witten's approach of Morse inequalities. He showed that, given a Morse function $g$ on a CW-complex $X$, there is always an equivalent Morse function $g$ that is self-indexing flat, i.e. its value is given by the dimension of the corresponding cell.  It turns out that, by considering the particular case of finite graphs, the computations of the deformed Laplacian and Morse-Witten complex are explicit, regardless whether the function is self-indexing or not.  In particular we deduce that the height function on trees (Section \ref{boundary}) gives rise to the correct Morse complex, after the deformation procedure.\\

The main idea can be summarized as follows: after introducing the supersymmetric version of quantum mechanics of graphs in terms of the graph Laplacian, we  describe a version of Morse theory on graphs first by ordering the set of edges and vertices of $\Gamma$ by declaring each vertex lesser than each edge of which it is an endpoint. With respect to this ordering, a \textit{discrete Morse function} is a real-valued function $f$ on the set of edges and vertices of $\Gamma$ such that for all $\sigma\in\Gamma$,
\begin{align*} \#\{\tau>\sigma | f(\tau) \leq f(\sigma)\} \leq 1 \\
\#\{\tau<\sigma | f(\tau) \leq f(\sigma)\} \leq 1.
\end{align*}

This means that a Morse function changes its value when there is a change in the dimension of the cells (0-cells being the vertices and 1-cells being the edges).
A \textit{critical point} $\sigma$ of $f$ is one for which the two sets above are empty.\\

The Morse inequalities state that Betti numbers $h_0$ and $h_1$ are bounded by the number of critical vertices and critical edges respectively. Using Theorem \ref{hodge-energy},  we arrive to these inequalities, drawing inspiration from work of Witten \cite{W}. The idea is as follows. Deform the supersymmetric Laplacian $\Delta$ using the Morse function $f$ with real parameter $s$ by taking boundary operator $d_s = \exp(fs)d\exp(-fs)$ and coboundary $d_s^* = \exp(fs)d^*\exp(-fs)$. The Hodge theorems still hold for the deformed Laplacian $$\Delta_s = d_s^*d_s + d_sd_s^*,$$ and after taking a limit $s\rightarrow \infty$, there is an energy level $a$ for which the cutoff complex $C^\bullet_a$ approaches the Morse complex as $s$ approaches infinity.\\

This combinatorial approach, based on the linear algebraic properties of the deformed Laplacian and incidence matrices, gives an intuitive interpretation of Witten's proof of Morse inequalities. In our description, the analytical issues of explicitly describing the spectra of deformed Laplacian operators (for which Witten requires to approximate the operators around the critical points by using the Morse coordinates) do not exist, since the operators are finite dimensional.

\section{Graph quantum mechanics}
In this section we introduce the combinatorial version of quantum mechanics for finite graphs. 
\subsection{The graph Laplacian}
For the purpose of this paper we consider finite graphs $\Gamma=(V, E)$, i.e. a finite set $V$ or vertices and a finite set $E$ of edges $e=(v_i, v_j)$. We will distinguish between unoriented and oriented graphs when necessary. 

\begin{figure}[h]
\begin{center}
	\begin{tikzpicture}
    \tikzset{
    mid arrow/.style={
    decoration={markings,mark=at position 0.5 with {\arrow[scale = 2]{>}}},
    postaction={decorate},
    shorten >=0.4pt}}
	\path (0,0) coordinate (X4); \fill (X4) circle (3pt);
	\path (0,1) coordinate (X1); \fill (X1) circle (3pt);
	\path (1,1) coordinate (X2); \fill (X2) circle (3pt);
	\path (1,0) coordinate (X5); \fill (X5) circle (3pt);
    \path (2,0) coordinate (X6); \fill (X6) circle (3pt);
    \path (2,1) coordinate (X3); \fill (X3) circle (3pt);
	\node[above left] at (X1) {$v_{1}$}; 
	\node[above right] at (X2) {$v_{2}$}; 
	\node[above right] at (X3) {$v_{3}$};
	\node[below left] at (X4) {$v_{4}$};
    \node[below] at (X5) {$v_{5}$};
    \node[below right] at (X6) {$v_{6}$};
	\draw (X1) -- (X2) node[midway,above]{$e_{1}$};
    \draw (X1) -- (X4) node[midway,left]{$e_{2}$};
    \draw (X2) -- (X4) node[midway,right]{$e_{3}$};
	\draw (X4) -- (X5) node[midway,below]{$e_{4}$};
    \draw (X5) -- (X6) node[midway,below]{$e_{5}$};
    \draw (X6) -- (X3) node[midway,right]{$e_{6}$};
    \draw (X3) -- (X5) node[midway,left]{$e_{7}$};
	\end{tikzpicture}
    \end{center}
    \caption{A connected unoriented graph, with two independent closed paths}
\label{Figure 1}\end{figure}
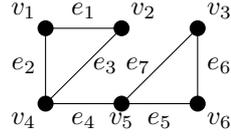

\begin{definition}
The adjacency matrix $A_{\Gamma}$ (or simply $A$) of the graph $\Gamma$ is
$$
A_{\Gamma}(i,j) = \left\{ \begin{array}{rl}
 1 &\mbox{ if $(v_i,v_j) \in E$} \\
  0 &\mbox{ otherwise}
       \end{array} \right.
$$
\end{definition} 

\begin{definition}
The valence matrix $val_{\Gamma}$ of the graph $\Gamma$ is the diagonal matrix such that the entry $(i,i)$ is the number of neighbors of the vertex $v_i$.
\end{definition}

\begin{definition}\label{even}
The even graph Laplacian $\Delta_{+, \Gamma}$ is defined by 
\begin{equation}
\Delta_{+, \Gamma}= val_{\Gamma}-A_{\Gamma}
\end{equation}
\end{definition}

\subsection{Orientation on graphs}
In order to define the incidence matrix of $\Gamma$, we choose an orientation, that is, a particular order of the pairs $(v_i,v_j)$.
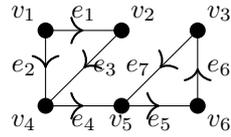
\begin{figure}[h]
\begin{center}
	\begin{tikzpicture}
    \tikzset{
    mid arrow/.style={
    decoration={markings,mark=at position 0.5 with {\arrow[scale = 2]{>}}},
    postaction={decorate},
    shorten >=0.4pt}}
	\path (0,0) coordinate (X4); \fill (X4) circle (3pt);
	\path (0,1) coordinate (X1); \fill (X1) circle (3pt);
	\path (1,1) coordinate (X2); \fill (X2) circle (3pt);
	\path (1,0) coordinate (X5); \fill (X5) circle (3pt);
    \path (2,0) coordinate (X6); \fill (X6) circle (3pt);
    \path (2,1) coordinate (X3); \fill (X3) circle (3pt);
	\node[above left] at (X1) {$v_{1}$}; 
	\node[above right] at (X2) {$v_{2}$}; 
	\node[above right] at (X3) {$v_{3}$};
	\node[below left] at (X4) {$v_{4}$};
    \node[below] at (X5) {$v_{5}$};
    \node[below right] at (X6) {$v_{6}$};
	\draw[mid arrow] (X1) -- (X2) node[midway,above]{$e_{1}$};
    \draw[mid arrow] (X1) -- (X4) node[midway,left]{$e_{2}$};
    \draw[mid arrow] (X2) -- (X4) node[midway,right]{$e_{3}$};
	\draw[mid arrow] (X4) -- (X5) node[midway,below]{$e_{4}$};
    \draw[mid arrow] (X5) -- (X6) node[midway,below]{$e_{5}$};
    \draw[mid arrow] (X6) -- (X3) node[midway,right]{$e_{6}$};
    \draw[mid arrow] (X3) -- (X5) node[midway,left]{$e_{7}$};
	\end{tikzpicture}
    \end{center}
    \caption{An oriented labeled graph}
\label{Figure 2}\end{figure}
\begin{definition}

Let $\Gamma$ be an oriented graph. The incidence matrix $I_{\Gamma}$ (or simply $I$) is a $\vert V \vert \times \vert E \vert$-matrix defined by

$$
I_{\Gamma}(k,l) = \left\{ \begin{array}{rl}
 -1 &\mbox{ if $e_l$ starts at $v_k$} \\
 1 &\mbox{if $e_l$ ends at $v_k$ } \\
  0 &\mbox{ otherwise}
       \end{array} \right.
$$
\end{definition} 
The following relationship between the even Laplacian and the incidence matrix follows form the combinatorial description of the even Laplacian.

\begin{proposition} \label{Lap}
The even graph Laplacian $\Delta_{+,\Gamma}$ can be written in terms of the incidence matrix, as follows:
$\Delta_{+,\Gamma}=I_{\Gamma} I^{*}_{\Gamma},$
\end{proposition}
Note that from Definition \ref{even} it follows that the even Laplacian is independent of the orientation.
Based on proposition \ref{Lap}, we have the following definition of the odd Laplacian, which can be regarded as an operator on functions on edges.
\begin{definition}
The odd Laplacian $\Delta_{-, \Gamma}$ is defined by 
\begin{equation}
\Delta_{-,\Gamma}=I^{*}_{\Gamma} I_{\Gamma}.
\end{equation}
\end{definition}

\subsection{The state evolution}

\begin{definition}
An even quantum state $\Psi_{+}$ on a graph $\Gamma$ is a complex-valued function on the vertices $\Gamma_0$, that is, $\Psi_{+} \in \mathbb C^{\vert V \vert}$. Similarly, an odd state $\Psi_{-}$ is an element of $\mathbb C^{\vert E \vert}$.
\end{definition}
The quantum theory is given by the Schr\"{o}dinger equation
\beq \frac{\partial}{\partial t}\Psi_{+,t} = -\Delta_{+} \Psi_{+,t} 
\qeb which is solved by
\beq \exp(-\Delta_{+} t) \Psi_{+,0}
\qeb
We denote $Z(t) = \exp(-\Delta_{+, t})$, the (even) \textit{partition function} of $\Gamma$. Indeed, if $\Gamma$ is regular, then
\beq Z(t)(i, j) = \exp(tA)\exp(-\val t)(i, j) = \sum_{\gamma: i \rightarrow j} \frac{t^{|\gamma|}}{|\gamma|!}e^{-\val t},
\qeb
which is an integral over a space of paths with measure $\frac{t^{|\gamma|}}{|\gamma|!}$ and integrand $e^{-\val t}$. The ``action'' on a path is the sum of the valences over the vertices it traverses. $Z(t)$ is therefore a discretization of the Feynman path integral. Furthermore,
\beq \frac{d^k Z}{dt^k}\Big|_0(i, j)
\qeb
gives a signed count of \textit{generalized walks} \cite{Yu} of length $k$ from vertex $i$ to $j$: a sequence $$(v_1, e_1), (v_2, e_2)..., (v_k, e_k)$$ of pairs of vertices and edges in which $v_j$ and $v_{j+1}$ are endpoints to $e_j$ (not necessarily distinct) for all $j$. In other words, a new path is a sequence of vertices which may traverse an edge while remaining stationary at a vertex.  The precise combinatorial interpretation of the partition function can be found in \cite{delVeccio, Mnev2, Yu}. The sign of such a path is determined by the number of $j$ such that $v_{j+1}\neq v_j$.

\subsection{The supersymmetric version}

  In the supersymmetric theory, we extend $\Delta$ to the entire simplicial cochain complex $C^\bullet$ of $\Gamma$ with $\Delta := \Delta_{+} \oplus \Delta_{-}$, where $\Delta_{-}$ and $\Delta_{+}$ operate on edges and vertices respectively. The entries of ${\Delta-}^k$ give signed counts of another type of special path: sequences of edges $e_1, e_2, ... , e_k$ such that $e_j$ is adjacent to $e_{j+1}$, with the sign determined by the number of $j$ such that $e_j$ meets $e_{j+1}$ with opposite orientation. For further details see, e.g. \cite{delVeccio, Yu}.\\
  \subsubsection{Graph Hodge Theory}
There is a close relationship between the graph Laplacian and the topology of the graph. More precisely,  we have the following relationship between $\Delta$ and the cohomology groups of $\Gamma$

\begin{proposition}\label{evenk}
The dimension of $\ker(\Delta_{+})$ is the number of connected components of $\Gamma$.
\end{proposition}
\begin{proof}
By Lemma \ref{kernel} in Appendix A, $\ker(\Delta_+)=\ker (I)$. Now, if the vertices $\{v_{\alpha_1}, v_{\alpha_2}, \cdots, v_{\alpha_k} \}$ are all the elements of a connected component of $\Gamma$, then the state 
$$
\Psi_{+}(v) = \left\{ \begin{array}{rl}
 1 &\mbox{ if $v=v_{\alpha_j}, 1\leq j \leq k$} \\
  0 &\mbox{ otherwise}
       \end{array} \right.
$$
is a generator of $\ker(I)$.
\end{proof}

\begin{proposition}
The dimension of $\ker(\Delta_{-})$ is the number of independent cycles of $\Gamma$. \end{proposition}
\begin{proof}
Once again, by Lemma \ref{kernel}, $\ker (\Delta_{-})=\ker (I^*)$. Now, the combinatorial interpretation of the elements of the kernel of $I^*$ is in terms of closed paths and this can be interpreted in terms of closed currents.  More precisely, one may think as a state $\Psi_{-}$ as a current, i.e. a function taking values on edges,  obeying  
Kirchhof's  first law: at each vertex, in-
and outgoing currents balance. Thus, a current achieves balance if and only if the current is assigned to a closed path.
\end{proof}

These two propositions lead to the following 
\begin{theorem}
The cohomology groups of $\Gamma$ can be calculated as
\beq
H^0(\Gamma, \CC) = \ker(I^*) = \ker(\Delta_{+})
\qeb
and
\beq H^1(\Gamma, \CC) = \ker(I) = \ker(\Delta_{-}).
\qeb
\end{theorem}
Thus $\Delta$ satisfies a ``discrete'' Hodge theorem.

We can further simplify the calculation of cohomology by considering energy cut-offs. Since $\Delta$ is symmetric, it is diagonalizable and its eigenvalues are real (non-negative, in fact), so the cochain complex decomposes as

\beq C^\bullet = \bigoplus_{\lambda\geq 0}E_\lambda
\qeb
where $E_\lambda$ is the eigenspace of $\Delta$ corresponding to $\lambda$. Given an energy $a\geq 0$, let
\beq C_a^\bullet = \bigoplus_{\lambda \leq a} E_\lambda
\qeb
be the sub-cochain complex of $C^\bullet$ consisting of eigenspaces of energy lower than $a$.

\begin{theorem}\label{hodge-energy}[Energy cut-off]
$H^*(C_a^\bullet) = H^*(C^\bullet)$
\end{theorem}
\begin{proof}
It is clear that $C_a^0$ contains $\ker \Delta_{+}$, for all $a$, thus $\ker I^*\vert _{C_a^0}=\ker \Delta_{+}=H^0(\Gamma)$. For $H^1$, it follows from the fact that $\coker (I^*)$ is contained in $C_a^1$, for all $a$, and from Lemma \ref{kernel}.
\end{proof}

\begin{remark}
In other words, the cohomology of $\Gamma$ can be calculated by considering subcomplexes of lower energy!
\end{remark}

\section{Morse theory}
In this section we describe a version of Morse theory on graphs by ordering the set of edges and vertices of $\Gamma$ by declaring each vertex lesser than each edge of which it is an endpoint. With respect to this ordering, we consider special states which change value (at least in one direction) when there is a change of the dimension (from 0 to 1 and vice-versa). 
More precisely we have the following definition, originally due to Forman \cite{Forman}:
\begin{definition}\label{Morsef}
a \textit{discrete Morse function} is a real-valued function $f$ on the set of edges and vertices of $\Gamma$ such that for all $\sigma\in\Gamma$,
\begin{eqnarray} 
\#\{\tau>\sigma | f(\tau) \leq f(\sigma)\} \leq 1 \label{MC1} \\
\#\{\tau<\sigma | f(\tau) \geq f(\sigma)\} \leq 1 \label{MC2}
\end{eqnarray}
\end{definition}

\begin{definition}
A \textit{critical cell} (vertex or edge) $\sigma$  of $f$ is one for which the two sets above are empty. We denote by $c_0(f)$ the number of critical vertices and by $c_1(f)$ the number of critical edges.
\end{definition}

The following lemma follows directly from definition \ref{Morsef}  (see e.g. \cite{Forman}, Lemma 2.4), and it will be useful when we will describe discrete gradient fields.
\begin{lemma}
If $\Gamma$ is a finite graph with discrete Morse function $f$, for every vertex or edge $\sigma$ either one of the following conditions hold:
\begin{enumerate}
\item[i.] \[ \#\{\tau>\sigma | f(\tau) \leq f(\sigma)\} =0,\] 
\item[ii.] \[\#\{\tau<\sigma | f(\tau) \geq f(\sigma)\} =0.\]

\end{enumerate}
\end{lemma}
As we have said before, our main goal is to prove the following theorem.

\begin{theorem}\label{Morse} (Graph Morse inequalities)
Let $h_0$ and $h_1$ be the Betti numbers of $\Gamma$. Then $h_0 \leq c_0(f)$ and $h_1 \leq c_1(f)$, for every Morse function $f$.
\end{theorem}
\begin{proof}
The strategy is as follows. By Using Theorem \ref{hodge-energy} above, we follow Witten's approach to Morse inequalities for Riemannian manifolds \cite{W}, via deformation of the supersymmetric Laplacian.  The precise idea is as follows. We deform the supersymmetric Laplacian $\Delta$ using the Morse function $f$ with a real parameter $s$ by  deforming the boundary and coboundary operators $d$ and $d^*$.  The Hodge theorems still hold for the deformed Laplacian $\Delta_s$, and, taking a limit $s\rightarrow \infty$, there is an energy level $a$ for which the cutoff complex $C^\bullet_a$ approaches the Morse complex as $s$ approaches infinity.\\
Now, let us start by deforming the boundary operator. 
\begin{definition} The deformed boundary and coboundary operators $d_s$ and $d_s^*$ are given by   
$$d_s = \exp(fs)d\exp(-fs), \, d_s^* = \exp(-fs)d^*\exp(fs).$$ 
\end{definition}
\begin{remark} Note that the $\exp{(fs)}$ and $\exp{(-fs)}$ are represented by matrices of, a priori, different dimensions. For the deformed boundary operator, $\exp(-fs)$ is a diagonal $\vert E \vert \times \vert E \vert $-real matrix, whereas $\exp(fs)$ is a diagonal $\vert V \vert \times \vert V \vert $-real matrix. In the coboundary case, the situation is reversed.
\end{remark}
\begin{definition} The deformed Laplacian $\Delta_s$ is defined by
$\Delta_s = d_s^*d_s + d_sd_s^*$
\end{definition}
Therefore we can define the deformed co-chain complex $C_s^{\bullet}$ as:
\begin{equation}\label{seq2}
\begin{tikzcd}
    0 \arrow{r}& \mathbb C^{\vert V \vert}\arrow{r}{d_s^*}&\mathbb C^{\vert E \vert }\arrow{r}&0.
\end{tikzcd}
\end{equation}
Similarly we can define the cut-off cochain complexes $C_{s,a}^{\bullet}$.

If we denote by $H^{\bullet}_s(\Gamma)$ the cohomology of the cochain $C_{s,a}^{\bullet}$, the following proposition follows from lemma \ref{conjugacy} in Appendix A, since the matrices $\exp(-fs)$ and $\exp(fs)$ are both invertible.
\begin{proposition} [deformed Energy cut-off]
\[H^{\bullet}_{s,a}(\Gamma)=H^{\bullet}(\Gamma).\]
\end{proposition}

Now, if we take the limit $s \to \infty$ the matrices $\Delta_{\pm, \infty}$ become quite simple, and their kernels become independent on $s$, only they only depend on the critical cells. Explicitly we have the following description:

\begin{proposition} \label{limit}
The matrices $\Delta_{+,\infty}$ and $\Delta_{-,\infty}$  with entries 0 and 1, and the number of zero columns is the number of corresponding critical cells.
\end{proposition}
\proof{ The general entries of the deformed boundary operator have the form $\exp(ks)$, where $k$ is the value of the Morse function on the respective cell. Thus, the graph Laplacian will have an entry of the form $\exp(qs)$, with $q$ nonzero, if and only the Morse value of a cell and its incident cell is different. Thus when $s\to \infty$ , the 1 entries  will occur exactly when there is a non critical cell, for which the value of the cell and one of the incident cells is the same.
}

From proposition \ref{limit} we conclude that the dimension of $\ker(\Delta_{+,\infty})$ is $c_0(f)$ and that the dimension of $\ker(\Delta_{-,\infty})$ is $c_1(f)$. Therefore, if $a$ is arbitrarily small, the energy cut-off produces a co-chain complex isomorphic to the Morse complex. This concludes the proof of Theorem \ref{Morse}.
\end{proof}
\section{Discrete gradient vector fields}
In \cite{W}, the super-symmetric interpretation of Morse inequalities is described in terms of \emph{instantons}, i.e. solutions of the differential equation
\begin{equation}\label{Instanton}
\frac{du(t)}{dt}=-\nabla f(u(t)), u(0)=q,
\end{equation}
where $q$ is a given non critical point, and with boundary conditions
\begin{eqnarray*}
&&\lim_{t\to \infty} u(t)=p\\
&&\lim_{t\to -\infty} u(t)=r,
\end{eqnarray*}
for which $p$ and $r$ are critical points of the Morse function $f$. 

The $CW$-decomposition for a well behaved Morse function (a so called Morse-Smale function, with suitable transversality conditions between the descending and ascending cells) comes equipped with an orientation, and a signed count of the number of solutions of Equation \ref{Instanton} gives the Morse differential for the Morse complex.\\
In the graph setting, there is a discrete analogue of a gradient vector field \cite{Forman}. It turns out that non critical cells always come in pairs. In order to see this, we observe that given a a non critical edge, it implies by definition that there exists an incident vertex with a nondecreasing value of the Morse function. In the same way, a non critical vertex has a adjacent edge such that the Morse value is nonincreasing. 

\begin{definition}
Let $f$ be a discrete Morse function on a graph $\Gamma$. The discrete gradient vector field of $f$, denoted by $\nabla f$, is the set of adjacent noncritical pairs $(v_{n_i
}, e_{n_i})$.
\end{definition}

Usually discrete gradient fields are represented graphically by arrows having  non critical vertices as tails and adjacent non critical edge as arrowheads, see e.g. Figure \ref{Figure 5}.
The following definition of gradient curves for a Morse function is the graph version of gradient paths given in \cite{FormanCW}:

\begin{definition}\label{curve}
A gradient curve between two  vertices  $\sigma_{\text{initial}}$ and $\sigma_{\text{final}}$ is a finite sequence
\[\gamma: \sigma_{\text{initial}}=\sigma_{0}, \tau_0, \sigma_1, \tau_1,\ldots, \tau_{k-1}, \sigma_k=\sigma_{\text{final}} \]
\end{definition}
such that the following conditions are satisfied
\begin{enumerate}
\item $\sigma_i < \tau_i$ and $\sigma_{i+1} < \tau_i$
\item $\sigma_i \neq \sigma_{i+1}$
\item $f(\sigma_i)\geq f(\tau_i)> f(\sigma_{i+1}).$
\end{enumerate}
We should interpret gradient curves as discrete solutions of Equation \ref{Instanton}.
In \cite{W}, each gradient curve has naturally equipped with a sign, so the Morse differential is obtained by counting the signed gradient curves among critical points of index differing by 1. In \cite{FormanCW}, the sign (or algebraic multiplicity) of a gradient curve is defined as follows.  Given an orientation on the vertices $\sigma _i$, the sign of a path $\gamma$, denoted by  $m(\gamma)$,  is said  to be $+1$ if the orientation on $\sigma_{\text{final}}$ agrees with the induced orientation on $\sigma_{\text{initial}}$, and is $-1$ otherwise.
Now, we can define the Morse differential $\tilde{\partial}$ from critical edges to critical vertices as follows. If $C^1, C^0$ denote the vector spaces generated by critical edges and vertices respectively, and an inner product $\langle \bullet , \bullet \rangle$ is chosen so the critical cells form an orthonormal basis, the linear operator
\begin{equation}\label{Diff}
\langle \tilde{\partial}\tau, \sigma \rangle= \sum_{\sigma_1, \tau} \langle \partial \tau, \sigma_1 \rangle \sum_{\gamma \in \Gamma(\sigma_1, \sigma)} m(\gamma)
\end{equation}
is clearly a differential, and furthermore it is the Morse differential \cite{FormanCW}.

\section{Examples}
\subsection{The simplest case} Let as consider a very simple graph, that is, a graph with two vertices and one edge: $V=\{v_1,v_2\}$ and $E=\{e_1=(v_1,v_2)\}$. 
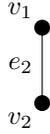
\begin{figure}[h]
\begin{center}
	\begin{tikzpicture}
    \tikzset{
    mid arrow/.style={
    decoration={markings,mark=at position 0.5 with {\arrow[scale = 2]{>}}},
    postaction={decorate},
    shorten >=0.4pt}}
	\path (0,0) coordinate (X4); \fill (X4) circle (3pt);
	\path (0,1) coordinate (X1); \fill (X1) circle (3pt);
		\node[above left] at (X1) {$v_{1}$}; 
		\node[below left] at (X4) {$v_{2}$}; 
 	\draw (X1) -- (X4) node[midway,left]{$e_{2}$};
	\end{tikzpicture}
    \end{center}
 \caption{The graph $K_2$.}
\label{Figure 3}
\end{figure}

With respect to the orientation given in Figure \ref{Figure 4} 
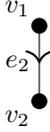
\begin{figure}[h]
\begin{center}
	\begin{tikzpicture}
    \tikzset{
    mid arrow/.style={
    decoration={markings,mark=at position 0.5 with {\arrow[scale = 2]{>}}},
    postaction={decorate},
    shorten >=0.4pt}}
	\path (0,0) coordinate (X4); \fill (X4) circle (3pt);
	\path (0,1) coordinate (X1); \fill (X1) circle (3pt);
		\node[above left] at (X1) {$v_{1}$}; 
		\node[below left] at (X4) {$v_{2}$}; 
 	\draw [mid arrow] (X1) -- (X4) node[midway,left]{$e_{2}$};
	\end{tikzpicture}
    \end{center}
    \caption{Oriented $K_2$.} 
\label{Figure 4}
\end{figure}
 
    %\caption{The complete graph $K2$}

we get the following matrices:
\begin{equation}
I_{\Gamma}=
\begin{pmatrix}
 -1\\1
\end{pmatrix}, \Delta_{+,\Gamma}=\begin{pmatrix} 1&-1\\-1&1\end{pmatrix}, \Delta_{-,\Gamma}=[2].
\end{equation}
Let us consider the following Morse function $f$ on $\Gamma$:

\begin{figure}[h]
\begin{center}
	\begin{tikzpicture}
    \tikzset{
    mid arrow/.style={
    decoration={markings,mark=at position 0.5 with {\arrow[scale = 2]{>}}},
    postaction={decorate},
    shorten >=0.4pt}}
	\path (0,0) coordinate (X4); \fill (X4) circle (3pt);
	\path (0,1) coordinate (X1); \fill (X1) circle (3pt);
		\node[above left] at (X1) {$1$}; 
		\node[below left] at (X4) {$0$}; 
 	\draw (X1) -- (X4) node[midway,left]{$1$};
	\end{tikzpicture}
    \end{center}
 \caption{The Morse function $f$ on $K_2$.}
\label{Figure 5}
\end{figure}
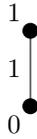
It is easy to check that $c_0(f)=1$ (the vertex $v_2$ is critical) and that $c_1(f)=0$ (there are no critical edges).
Now, the deformed boundary $d_s$ is 
\begin{equation}
d_s=I_{\Gamma,s}=(\exp{(sf)})I_{\Gamma}(\exp{(-sf)})= \begin{bmatrix}\exp(s)&0\\0&1  \end{bmatrix} \begin{bmatrix}-1\\1  \end{bmatrix} [\exp(-s)]=\begin{bmatrix} 1\\-\exp(-s) \end{bmatrix}.
\end{equation}
Therefore, the deformed even Laplacian $\Delta_{+,s}$ is
\begin{equation}
\Delta_{+,s}=d_sd_s^*= \begin{bmatrix} 1&-\exp{(-s)}\\-\exp{(-s)}&\exp{(-2s)} \end{bmatrix}
\end{equation}
and the odd Laplacian $\Delta_{-,s}$ is
\begin{equation}
\Delta_{-,s}=d_s^*d_s= [1+\exp{(-2s)}].
\end{equation}
Therefore,
\begin{equation}
\Delta_{+,\infty}= \begin{bmatrix} 1&0\\0&0 \end{bmatrix}, \Delta_{-,\infty}=[1].
\end{equation}
It can be easily checked that $\dim(\ker{(\Delta_{+,\infty})})=1=c_0(f)$ and that $\dim(\ker{(\Delta_{+,\infty})})=0=c_1(f)$. More precisely,
\begin{eqnarray*}
\ker{(\Delta_{+,\infty})}&=&\langle v_2 \rangle\\
\ker{(\Delta_{-,\infty})}&=&\langle 0 \rangle.
\end{eqnarray*}

\subsection{The triangle}
We illustrate the case in which we have two different Morse functions, one of each achieving sharpness of the Morse inequalities.
Let as consider the triangle graph $K_3$ with $V=\{v_1,v_2, v_3\}$ and $E=\{e_1=(v_1,v_2), (v_1,v_3),(v_2,v_3)\}$. \\

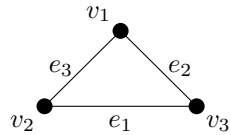
\begin{figure}[h]
\begin{center}
	\begin{tikzpicture}
    \tikzset{
    mid arrow/.style={
    decoration={markings,mark=at position 0.5 with {\arrow[scale = 2]{>}}},
    postaction={decorate},
    shorten >=0.4pt}}
	\path (0,0) coordinate (X4); \fill (X4) circle (3pt);
	\path (2,0) coordinate (X5); \fill (X5) circle (3pt);
         \path (1,1) coordinate (X1); \fill (X1) circle (3pt);
		\node[above left] at (X1) {$v_1$}; 
		\node[below left] at (X4) {$v_2$}; 
		\node[below right] at (X5) {$v_3$}; 

 	\draw (X1) -- (X4) node[midway,left]{$e_3$};
 	\draw (X1) -- (X5) node[midway,right]{$e_2$};
 	\draw (X4) -- (X5) node[midway,below]{$e_1$};
	
	\end{tikzpicture}
    \end{center}
 \caption{The triangle $K_3$.}
\label{Figure 5}
\end{figure}

\newpage
With respect to the orientation given in Figure \ref{Figure 7}

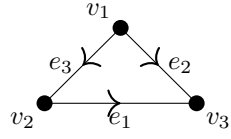
\begin{figure}[h]
\begin{center}
	\begin{tikzpicture}
    \tikzset{
    mid arrow/.style={
    decoration={markings,mark=at position 0.5 with {\arrow[scale = 2]{>}}},
    postaction={decorate},
    shorten >=0.4pt}}
	\path (0,0) coordinate (X4); \fill (X4) circle (3pt);
	\path (2,0) coordinate (X5); \fill (X5) circle (3pt);
         \path (1,1) coordinate (X1); \fill (X1) circle (3pt);
		\node[above left] at (X1) {$v_1$}; 
		\node[below left] at (X4) {$v_2$}; 
		\node[below right] at (X5) {$v_3$}; 

 	\draw  [mid arrow](X1) -- (X4) node[midway,left]{$e_3$};
 	\draw [mid arrow] (X1) -- (X5) node[midway,right]{$e_2$};
 	\draw  [mid arrow](X4) -- (X5) node[midway,below]{$e_1$};
	
	\end{tikzpicture}
    \end{center}
 \caption{An orientation on $K_3$.}
\label{Figure 7}
\end{figure}
 
    %\caption{The complete graph $K2$}

we get the following incidence matrix:
\begin{equation}
I_{\Gamma}=
\begin{pmatrix}
 0&-1&-1\\
 -1&0&1\\
 1&1&0
\end{pmatrix}.
\end{equation}

Now consider the following Morse function on $K_3$:

\begin{figure}[h]
\begin{center}
	\begin{tikzpicture}
    \tikzset{
    mid arrow/.style={
    decoration={markings,mark=at position 0.5 with {\arrow[scale = 2]{>}}},
    postaction={decorate},
    shorten >=0.4pt}}
	\path (0,0) coordinate (X4); \fill (X4) circle (3pt);
	\path (2,0) coordinate (X5); \fill (X5) circle (3pt);
         \path (1,1) coordinate (X1); \fill (X1) circle (3pt);
		\node[above left] at (X1) {$1$}; 
		\node[below left] at (X4) {$0$}; 
		\node[below right] at (X5) {$1$}; 

 	\draw [mid arrow] (X1) -- (X4) node[midway,left]{$1$};
 	\draw (X1) -- (X5) node[midway,right]{$2$};
	 \draw [mid arrow] (X5) -- (X4) node[midway,above]{$1$};

	\end{tikzpicture}
    \end{center}
 \caption{The Morse function $f$ on $K_3$ and the corresponding gradient vector field $\nabla f$.}
\label{Figure 5}
\end{figure}
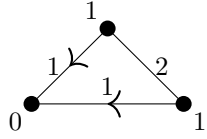

It is easy to check that $c_0(f)=1$ (the vertex $v_2$ is critical) and that $c_1(f)=1$ (the edge $e_2$ is critical).
Now, the deformed boundary $d^*_s$ is 
\begin{equation}
d_s= \begin{bmatrix} \exp(s)&0&0\\0&1&0\\0&0&\exp(s)  \end{bmatrix}\begin{bmatrix}
 0&-1&-1\\
 -1&0&1\\
 1&1&0
\end{bmatrix}\begin{bmatrix} \exp(-s)&0&0\\0&\exp(-2s)&0\\0&0&\exp(-s)  \end{bmatrix} =\begin{bmatrix} 0&-\exp(-s)&-1\\-\exp(-s)&0&\exp(-s)\\1&\exp(-s)&0 \end{bmatrix}.
\end{equation}
Therefore, the deformed even Laplacian $\Delta_{+,s}$ is
\begin{equation}
\Delta_{+,s}=d_sd_s^*= \begin{bmatrix} 1+\exp(-2s)&-\exp(-s)&-\exp(-2s)\\-\exp(-s)&2\exp(-2s)&-\exp(-s\\-\exp(-2s)&-\exp(-s)&1+\exp(-2s) \end{bmatrix}
\end{equation}
and the odd Laplacian $\Delta_{-,s}$ is
\begin{equation}
\Delta_{-,s}=d_s^*d_s= \begin{bmatrix} 1+\exp(-2s)&\exp(-s)&-\exp(-2s)\\\exp(-s)&2\exp(-2s)&-\exp(-s\\-\exp(-2s)&-\exp(-s)&1+\exp(-2s) \end{bmatrix}
\end{equation}
Therefore,
\begin{equation}
\Delta_{+,\infty}= \begin{bmatrix} 1&0&0\\0&0&0\\0&0&1 \end{bmatrix}= \Delta_{-,\infty}.
\end{equation}
It can be easily checked that $\dim(\ker{(\Delta_{+,\infty})})=1=c_0(f)=c_1(f)=\dim(\ker{(\Delta_{-,\infty})})$ and that 
\begin{eqnarray*}
\ker{(\Delta_{+,\infty})}&=&\langle v_2 \rangle\\
\ker{(\Delta_{-,\infty})}&=&\langle e_2 \rangle.
\end{eqnarray*}
As expected, this Morse function achieves the equality for the Morse inequalities.\\

On the other hand, we might have considered the  Morse function on $K_3$ given by figure \ref{M2}. For this function, all the vertices and edges are critical, thus $c_0(g)=c_1(g)=3$.

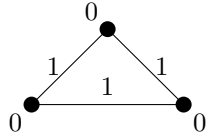
\begin{figure}[h]
\begin{center}
	\begin{tikzpicture}
    \tikzset{
    mid arrow/.style={
    decoration={markings,mark=at position 0.5 with {\arrow[scale = 2]{>}}},
    postaction={decorate},
    shorten >=0.4pt}}
	\path (0,0) coordinate (X4); \fill (X4) circle (3pt);
	\path (2,0) coordinate (X5); \fill (X5) circle (3pt);
         \path (1,1) coordinate (X1); \fill (X1) circle (3pt);
		\node[above left] at (X1) {$0$}; 
		\node[below left] at (X4) {$0$}; 
		\node[below right] at (X5) {$0$}; 

 	\draw (X1) -- (X4) node[midway,left]{$1$};
 	\draw (X1) -- (X5) node[midway,right]{$1$};
	 \draw (X4) -- (X5) node[midway,above]{$1$};

	\end{tikzpicture}
    \end{center}
 \caption{The Morse function $g$ on $K_3$.}
\label{M2}
\end{figure}

Now, the deformed boundary $d^*_s$ is 
\begin{equation}
d^*_s= \begin{bmatrix} 1&0&0\\0&1&0\\0&0&1  \end{bmatrix}\begin{bmatrix}
 0&-1&-1\\
 -1&0&1\\
 1&1&0
\end{bmatrix}\begin{bmatrix} \exp(-s)&0&0\\0&\exp(-s)&0\\0&0&\exp(-s)  \end{bmatrix} =\begin{bmatrix} 0&-\exp(-s)&-\exp(-s)\\-\exp(-s)&0&\exp(-s)\\\exp(-s)&\exp(-s)&0 \end{bmatrix}.
\end{equation}
Therefore, the deformed even Laplacian $\Delta_{+,s}$ is
\begin{equation}
\Delta_{+,s}=d_s^*d_s= \begin{bmatrix} 2\exp(-2s)&-\exp(-2s)&-\exp(-2s)\\-\exp(-2s)&2\exp(-2s)&-\exp(-2s)\\-\exp(-2s)&-\exp(-2s)&2\exp(-2s) \end{bmatrix}
\end{equation}
and the odd Laplacian $\Delta_{-,s}$ is
\begin{equation}
\Delta_{+,s}=d_s^*d_s= \begin{bmatrix} 2\exp(-2s)&\exp(-2s)&-\exp(-2s)\\\exp(-2s)&2\exp(-2s)&-\exp(-2s)\\-\exp(-2s)&-\exp(-2s)&2\exp(-2s) \end{bmatrix}.
\end{equation}
Therefore,
\begin{equation}
\Delta_{+,\infty}= \begin{bmatrix} 0&0&0\\0&0&0\\0&0&0 \end{bmatrix}= \Delta_{-,\infty},
\end{equation}
thus
\begin{eqnarray*}
\ker{(\Delta_{+,\infty})}&=&\langle v_1, v_2, v_3 \rangle\\
\ker{(\Delta_{-,\infty})}&=&\langle e_1, e_2, e_3 \rangle.
\end{eqnarray*}

\section{Morse function of a tree and the boundary map}\label{boundary}

Given a vertex $v$ in $\Gamma$, let $v'$ be the critical vertex obtained by flowing along the gradient field $\nabla f$ of $f$. For an edge $e$ in $\Gamma$ with endpoints $v_0$ and $v_1$, the Morse boundary map, defined on critical edges, is given by

\begin{equation}
	e \mapsto v_1' - v_0'.
\end{equation}

Let $T$ be a spanning tree of $\Gamma$ and $v_r$ a vertex in $T$, the \textit{root}. Then it is an easy observation that the height function $h$ defined by:
\begin{eqnarray*}
	h(v) &=& \text{ edge distance from $v_r$}\\
	h(e)&=& \text{max}(h(v_0),h(v_1)) \text{\hspace{1mm}if $e$ belongs to $T$}\\
	h(e)&=& \text{max}(h(v_0),h(v_1))+1 \text{\hspace{1mm}otherwise},
\end{eqnarray*}
is Morse.

The following proposition justifies the fact that we recover the Morse complex for such functions.

\begin{proposition}
The boundary map is zero on critical edges of $h$.
\end{proposition}

\begin{proof}
The critical cells of $h$ consist of $v_0$ and all the edges not contained in $T$, and the boundary map is zero on critical edges since $v_0'=v_1'$.  Therefore, the boundary map coincides with the operator $\tilde{\partial }$ from Definition \ref{Diff}, and this implies that the boundary maps compute graph homology, see e.g \cite{FormanCW}.
\end{proof}

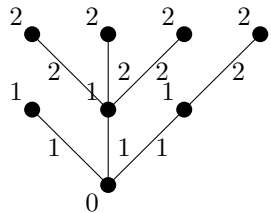
\begin{figure}[h]
\begin{center}
	\begin{tikzpicture}
    \tikzset{
    mid arrow/.style={
    decoration={markings,mark=at position 0.5 with {\arrow[scale = 2]{>}}},
    postaction={decorate},
    shorten >=0.4pt}}
	\path (1,0) coordinate (X1); \fill (X1) circle (3pt);
	\path (0,1) coordinate (X2); \fill (X2) circle (3pt);
	\path (1,1) coordinate (X3); \fill (X3) circle (3pt);
	\path (2,1) coordinate (X4); \fill (X4) circle (3pt);
	\path (0,2) coordinate (X5); \fill (X5) circle (3pt);
	\path (1,2) coordinate (X6); \fill (X6) circle (3pt);
	\path (2,2) coordinate (X7); \fill (X7) circle (3pt);
	\path (3,2) coordinate (X8); \fill (X8) circle (3pt);
		\node[below left] at (X1) {$0$}; 
		\node[above left] at (X4) {$1$}; 
		\node[above left] at (X5) {$2$}; 
		\node[above left] at (X2) {$1$}; 
		\node[above left] at (X3) {$1$}; 
		\node[above left] at (X6) {$2$}; 
		\node[above left] at (X7) {$2$}; 
		\node[above left] at (X8) {$2$}; 

 	\draw (X1) -- (X2) node[midway,left]{$1$};
 	\draw (X1) -- (X3) node[midway,right]{$1$};
	 \draw (X1) -- (X4) node[midway,right]{$1$};
	 \draw (X3) -- (X5) node[midway,left]{$2$};
         \draw (X3) -- (X6) node[midway,right]{$2$};
 	 \draw (X3) -- (X7) node[midway,right]{$2$};
	 \draw (X4) -- (X8) node[midway,right]{$2$};

	\end{tikzpicture}
    \end{center}
 \caption{The height function for a rooted tree. Note that the only critical cell of a rooted tree is the root, therefore there are no gradient curves and thus the Morse complex has trivially Morse homology with respect to the differential $\tilde \partial$.}
\label{M2}
\end{figure}

\section{Conclusion and perspectives}
We have re-proven Morse inequalities in the particular case of finite graphs, by using Witten's supersymmetric approach for quantum mechanics on Riemannian manifolds. The equality is achieved in both examples by a \emph{height type} Morse function, which can be defined for a spanning  tree in terms of the \emph{depth} of the tree, once a root is chosen. The remaining values of the Morse function can be chosen to be larger than the maximum of the corresponding edges, so the Morse conditions (\ref{MC1}) and (\ref{MC2}) are  satisfied. We conjecture that the sharpness of the equation is achieved by such functions in the general case of CW-complexes, for which there is a generalized notion of  a spanning tree and corresponding height function This will be part of an upcoming publication. We also intend to describe in detail how to use Witten's approach to derive Morse inequalities and define a discrete version of quantum mechanics in interesting higher dimensional examples such as real projective spaces and complexes of graphs with a monotone decreasing property \cite {Forman}.

\section*{Acknowledgements}
This research was conducted within the Illinois Geometry Lab (IGL) project \emph{Quantum Mechanics for Graphs and CW-Complexes}. 

I. C. thanks P. Mn\"ev for interesting discussions and for explaining the graph version of quantum mechanics and A. Cattaneo for useful comments on a preliminary version of the manuscript. I.C. and B.X thank Sarah Loeb, Michael Toriyama, Chengzheng Yu and Zhe Hu for useful discussions.
%{\color{red} I.C. and M.S. are thankful to the organisers of Poisson '14 for the support they were granted, as some of the results presented here have been worked out during the conference.}

\newpage
\appendix
\section{Linear algebra and the graph Laplacian}\label{App:lin}
The following are technical basic lemmas in linear algebra used throughout the paper, and they can be found in standard references for matrix linear algebra, such as \cite{Horn}.
\begin{lemma}\label{kernel}
Let $A$ be a matrix and let $A^*$ its adjoint. Then $\ker(A)=\ker(A^*A)$. 
\end{lemma}
\begin{proof}
It is clear that $\ker A\subseteq \ker(A^*A)$. For the other direction, if $\langle  \cdot , \cdot \rangle$ is the corresponding inner product, then, for each vector $v$ in $\ker (A^*A)$
\[ \vert \vert Av \vert\vert ^2=\langle Av, Av \rangle =\langle A^*Av, v \rangle= \langle 0, v \rangle = 0,\]
thus $v \in \ker (A)$.
\end{proof}

\begin{lemma} \label{spectrum}
The matrices $AA^*$ and $A^*A$ are both non-negative definite and their spectra coincide (modulo multiplicities).
\end{lemma}

\begin{lemma} \label{conjugacy}
Let $A$ be a $p \times q$ matrix,  let $X$ be an invertible $p \times p$- matrix and let $Y$ be an invertible $q \times q$-matrix. Then \[\ker (A)= \ker (XAY).\]
\end{lemma}

\end{document}